\documentclass[11pt, letterpaper]{article}

\usepackage{sketch}

\newcommand{\R}{\bbR}
\newcommand{\Z}{\bbZ}
\newcommand{\D}{\calD}
\newcommand{\half}{\frac{1}{2}}
\newcommand{\inorm}[1]{\left\Vert #1 \right\Vert}
\newcommand{\brackets}[1]{\left(#1\right)} 
\newcommand{\abs}[1]{\left\lvert #1 \right\rvert}

\begin{document}

\title{On Sketching the $q$ to $p$ Norms}
\author{Aditya Krishnan \footnotemark[1]
\and Sidhanth Mohanty \footnotemark[2]
\and David P. Woodruff \footnotemark[3]}

\footnotetext[1]{Computer Science Department, Carnegie Mellon University, Pittsburgh, PA, USA. \texttt{arkrishn@andrew.cmu.edu}}
\footnotetext[2]{Computer Science Department, Carnegie Mellon University, Pittsburgh, PA, USA. \texttt{sidhanthm96@gmail.com}}
\footnotetext[3]{Computer Science Department, Carnegie Mellon University, Pittsburgh, PA, USA. \texttt{dwoodruf@cs.cmu.edu}.
The author would like to acknowledge the support by the National Science Foundation under Grant No. CCF-1815840.}

\maketitle

\begin{abstract}
We initiate the study of data dimensionality reduction, or sketching, for the $q\to p$ norms.
Given an $n \times d$ matrix $A$, the $q\to p$ norm, denoted 
$\|A\|_{q \to p} = \sup_{x \in \mathbb{R}^d \backslash \vec{0}} \frac{\|Ax\|_p}{\|x\|_q}$, is a natural generalization of several
matrix and vector norms studied in the data stream and sketching models, with applications to datamining, hardness
of approximation, and oblivious routing. We say a distribution $S$ 
on random matrices
$L \in \mathbb{R}^{nd} \rightarrow \mathbb{R}^k$ is a $(k,\alpha)$-sketching family
if from $L(A)$, one can approximate $\|A\|_{q \to p}$ up to a factor $\alpha$ with constant probability. 
We provide upper and lower bounds on the sketching dimension $k$ for every $p, q \in [1, \infty]$, and in
a number of cases our bounds are tight. While we mostly focus on constant $\alpha$, we also consider large
approximation factors $\alpha$, as well as other variants of the problem such as when $A$ has low rank.
\end{abstract}

\setcounter{page}{0}
\thispagestyle{empty}
\newpage

\section{Introduction}
Data dimensionality reduction, or sketching, is a powerful technique by which one compresses a large dimensional object
to a much smaller representation, while preserving important structural information. Motivated by applications in streaming
and numerical linear algebra, the object is often a vector $x \in \mathbb{R}^n$ or a matrix $A \in \mathbb{R}^{n \times d}$. 
One of the most common forms of sketching is oblivious
sketching, whereby one chooses a random matrix $L$ from some distribution $S$, and compresses $x$ to $Lx$ or $A$ to $L(A)$. The latter
quantity $L(A)$ denotes a linear map from $\mathbb{R}^{nd}$, interpreting $A$ as an $nd$-dimensional vector, to an often much
lower dimensional space, say $\mathbb{R}^k$ for a value $k \ll nd$. 

Sketching has numerous applications. For example, in the data stream model, one sees additive updates $x_i \leftarrow x_i + \Delta$,
where the update indicates that $x_i$ should change from its old value by an additive $\Delta$. Given a sketch $L \cdot x$, one
can update it by replacing it with $L \cdot x + \Delta \cdot L_{*,i}$, where $L_{*,i}$ denotes the $i$-th column of $L$. Thus, it is easy
to maintain a sketch of a vector evolving in the streaming model. Similarly, in the matrix setting, given an update 
$A_{i,j} \leftarrow A_{i,j} + \Delta$, one can update $L(A)$ to $L(A) + \Delta L(e_{i,j})$, where $e_{i,j}$ denotes the matrix with
a single one in the $(i,j)$-th position, and is otherwise $0$. If $L$ is oblivious, that is, sampled from a distribution independent
of $x$ (or $A$ in the matrix case), then one can create $L$ without having to see the entire stream in advance. Other applications
include distributed computing, whereby a vector or matrix is partitioned across multiple servers. For instance, server $1$ might
have a vector $x^1$ and server $2$ a vector $x^2$. Given the sketches $Lx^1$ and $Lx^2$, by linearity 
one can combine them, using $L(x^1 + x^2) = Lx^1 + Lx^2$. In these applications it is important that the number $k$ of rows
of $L$ is small, since it is proporational to the memory required of the data stream algorithm, or the communication in a distributed
protocol. Here $k$ is referred to as the {\it sketching} dimension. 

Sketching vector norms is fairly well understood, and we have tight bounds up to logarithmic factors for estimating the $\ell_p$-norms 
$\|x\|_p = (\sum_i |x_i|^p)^{1/p}$ for every $p \in [1, \infty]$; for a sample of such work, see 
\cite{alon1996space,bar2002information,indyk2005optimal,indyk2006stable,kane2010exact,kane2011fast} for work in the related data
stream context, and \cite{pw12,andoni2013tight,lw13} for work specifically in the sketching model. 
Recently, there is work \cite{BBCKY17} characterizing the sketching
complexity of any symmetric norm on a vector $x$. A number of works have also looked at sketching
{\it matrix norms}. In particular, the Schatten $p$-norms $\|A\|_p = \left (\sum_{i=1}^{\textrm{rank}(A)} \sigma_i(A)^p \right)^{1/p}$ have gained
considerable attention. They have proven to be considerably harder to approximate than the vector $p$-norms, and understanding
their complexity has led to important algorithmic and lower bound techniques. A body of work has focused on 
understanding the complexity
of estimating matrix norms in the data stream model with $1$-pass
over the stream \cite{andoni2013eigenvalues,lw16}, as well as with
multiple passes \cite{b16}, the sketching model \cite{li2014sketching,lw17}, statistical models \cite{kv16,KO17}, as well as the general
RAM model \cite{MNSUW18,UCS16}. Dimensionality reduction in these norms
also has applications in quantum
computing \cite{w05,hms11}, and are studied in nearest neighbor search data structures \cite{a10}. 

\subsection{Our Contributions}
We consider the sketching complexity of a new family of norms, namely, the $p\to q$ norms of a matrix. 
A common quantity that arises in various applications is the amount by which a linear map $A$ ``stretches'' vectors. One way
to measure this quantity is the maximum singular value of $A$, which can be written as $\sup_{\|x\|_2=1}\|Ax\|_2$, and is just
the Schatten-$\infty$ norm, defined above. 
In this work we consider a different way of measuring this stretch, which considerably generalizes the operator norm.

For a linear operator $A$ from a normed space $\mathcal{X}$ to a normed space $\mathcal{Y}$, we define
$
\|A\|_{\mathcal{X}\to\mathcal{Y}}$ as $\sup_{\|x\|_{\mathcal{X}}=1}\|Ax\|_{\mathcal{Y}}
$.
Of specific interest to us is the case where $\mathcal{X}=\ell_q^d$ and $\mathcal{Y}=\ell_p^n$, and we denote the corresponding
norm of such an operator by $\|A\|_{q\to p}$.
Our objective is to study the sketching complexity of approximating this norm.

\begin{definition}[$(k,\alpha)$-sketching family]
Let $\mathcal{S}$ be a distribution over linear functions from $\R^{n\times d}$ to $\R^k$ and $f$ a function from $\R^k$ to $\R$. We call
$(\mathcal{S},f)$ a \textbf{$(k,\alpha)$-sketching family} for the $q\to p$ norm if for all $A\in\R^{n\times d}$,
$\Pr_{L\sim\mathcal{S}}\left[f(L(A)) \in \left(1/\alpha,\alpha \right) \|A\|_{q\to p}\right]\geq\frac{5}{6}$.
\end{definition}
We provide upper and lower bounds on $k$. The details of the specific results we have are described in \pref{sec:results}.

\subsection{Motivation}
This problem is well-studied in mathematics when $p = q$ as it simply corresponds to $p$-matrix norm 
estimation\footnote{See, e.g., \url{https://en.wikipedia.org/wiki/Matrix_norm}}.
An intriguing question is whether one can preserve $\|Ax\|_p$ in a lower-dimensional sketch space, given that the vectors
$x$ come from the unit ball of 
a smaller norm.

Apart from being mathematically interesting, this problem has a number of applications. The operator norm is a special case when $p = q = 2$. The operator norm can be
accurately estimated by any subspace embedding for $\ell_2$, discussed in detail in \cite{clarkson2013low}.
The dual of this norm is also the Schatten-$1$ norm, which has received considerable attention in the streaming model \cite{lw16,b16}. The $q\to p$ norm problem is a natural generalization of the operator norm problem, and when $p < 2$, may be more appropriate in the context of robust statistics, where it is known that the $p$ norm for $p < 2$ is less sensitive to outliers, see, e.g., Chapter 3 of \cite{w14} for a survey on robust regression, and \cite{swz17} for recent work on $\ell_1$-low rank approximation. 

The $2\to q$ norms arise in the hardness of approximation literature and an algorithm
for some instances of the problem was used to break the Khot-Vishnoi Unique Games
candidate hard instance \cite{khot2015unique}. 
Work by \cite{barak2012hypercontractivity} gives an algorithm
running in time $\exp(n^{2/p})$ for approximating $2\to p$ norms for all $p\geq 4$. These algorithms 
give a constant factor approximation when promised the $2\to p$ norm is in a certain range (depending on the operator norm) rather than 
providing a general estimate of the $2\to p$ norm. This same paper also discusses assumptions on the 
the NP-hardness and ETH hardness of approximating $2\to p$ norms. 
The work of \cite{brandao2015estimating}
extends that of \cite{barak2012hypercontractivity} to all $p\geq 2$.
The work of \cite{bhaskara2011approximating} gives a PTAS for computing $\|A\|_{q\to p}$ if $1\leq p\leq q$ and
$A$ has non-negative entries, and gives an application of this to the oblivious routing problem where congestion is
measured using the $\ell_p$ norm. The paper also shows that it is hard to approximate $\|A\|_{q\to p}$ within a
constant factor for general $A$, and general $p$ and $q$. Sketching may allow, for example, for reducing the original 
problem to a smaller instance of the same problem, which although may still involve exhaustive search, 
could give a faster concrete running time. 

The $1\to q$ norm turns out to be the maximum of the $q$-norm of the columns of $A$, which is related to the heavy hitters problems in data streams, e.g., the column with the largest $q$-norm may be the most significant or desirable in an application. Likewise, the $q\to \infty$ norms turn out to be the maximum of the $p$-norms of the rows of $A$, where $p$ is the dual norm to $q$, and therefore have similar heavy hitter applications. The $\infty\to q$ norm is maximized when $x \in \{-1,1\}^n$ and therefore includes the cut-norm as a special case, and is related to Grothendieck inequalities, see, e.g., \cite{bfv10,nrv14,bos17}. 

Our main motivation for studying the $p\to q$ norms comes from understanding and developing new techniques for this family of norms. Another family of norms that is well-studied in the data stream literature are the {\it cascaded norms}, which for an $n \times d$ matrix $A$ and parameters $p$ and $q$, are defined to be $(\sum_{i=1,\ldots, n} (\|A_{i,*}\|_p)^q)^{1/q}$, where $A_{i,*}$ denotes the $i$-th row of $A$. That is, we compute the $q$-norm of the vector of $p$-norms of the rows of $A$. This problem originated in \cite{cormode2005space} and has applications to mining multi-graphs; the following sequence of work established tight bounds up to logarithmic factors for every $p,q \in [1, \infty]$ \cite{jayram2009data,andoni2011streaming}. This line of work led to very new techniques; one highlight is the use of Poincar\'e inequalities in proving information complexity lower bounds, which has then been studied in a number of followup works \cite{ajp10,j13,akr15}.

\subsection{Our Results}\label{sec:results}
After establishing preliminary results and theorems in \pref{sec:preliminaries}, we give our results for constant and large approximation factors. Our main theorem is as follows. Here
 $\ell_{q^*}$ is the dual norm of $\ell_q$, that is, 
$1/q^* + 1/q = 1$ (when $q = 1$, $q^* = \infty$, and vice versa). 
\begin{theorem}\label{thm:summarytheorem}
For all matrices $A \in \R^{n \times n}$ with rank $r$ and real values $p, q \in [1, \infty]$, the table below gives upper and lower bounds on $k$ for a $(k, \Theta(1))$-sketching family of various $q \to p$ norms.

\begin{center}
\begin{tabular}{ |c|c|c|c|c|c| } 
\hline
\textbf{$q \to p$ Norm} & \textbf{$p^* \to q^*$ Norm} & \textbf{Upper Bound} & \textbf{Sec} & \textbf{Lower Bound} & \textbf{Sec} \\
\hline
$1 \to [1, 2]$ & $[2, \infty] \to \infty$ & $O(n\log n)$ & \ref{sec:1toPupper} & $\Omega(n)$ & \ref{sec:1toplower} \\ 
$1 \to [2, \infty]$ & $[1, 2] \to \infty$ & $O(n^{2-\frac{2}{p}}\log^2n)$ & \ref{sec:1toPupper} & $\Omega(n^{2-\frac{2}{p}})$ & \ref{sec:1toInftylower} \\ 
$[2, \infty] \to [1, 2]$ & $[2, \infty] \to [1, 2]$ & $O(n^2)$ & - & $\Omega(n^2)$ & \ref{sec:qlarge_to_psmall_lower} \\ 
$2 \to [2, \infty]$ & $[1, 2] \to 2$ & $O(\min\{n^{1-\frac{2}{p}}r^2\log n, n^2\})$ & \ref{sec:2toqupper} & $\Omega(\min\{n,n^{1-\frac{2}{p}}r\})$ & \ref{sec:ptoq_oneside} \\ 
$[1, 2] \to [1, 2]$ & $[2, \infty] \to [2, \infty]$ & $O(n^2)$ & - & $\Omega(\min\{n^{1-\frac{2}{q^*}}r,n\})$ & \ref{sec:ptoq_oneside} \\ 
$[1, 2] \to [2, \infty]$ & $[1, 2] \to [2, \infty]$ & $O(n^2)$ & - & $\Omega\left(\frac{n}{\log n}\right)$ & \ref{sec:qsmall_plarge} \\
\hline
\end{tabular}
\end{center}
\end{theorem}

The constant factor hidden in \pref{thm:summarytheorem} does not hold for all constants, the smallest constant it holds for varies depending on the specific values of $q, p$.

We also have several results for large approximation factors summarized in the theorem below. 
\begin{theorem}
There exists a $\left(O\left(\frac{n^2}{\alpha}\right),\alpha\right)$-sketching family for the $2\to p$ and $\infty\to p$ norm and a $\left(O\left(\frac{n^2}{\alpha^2}\right),\alpha\right)$-sketching family for the $q\to p$ norm for $q\geq 1$ and $1\leq p\leq 2$.
\end{theorem}

Our algorithms combine several insights, which we illustrate here 
in the case of the $2 \to p$ norm for $p \geq 2$ and when the rank of $A$ is $r$: 
(1) we show by duality that $\|A\|_{2 \to p}$
is the same as $\|A^T\|_{p^* \to 2}$, where $p^*$ 
satisfies $\frac{1}{p^*} + \frac{1}{p} = 1$ and is the dual norm to $p$. Although
the proof is elementary, this plays several key roles in our argument. 
Next, we (2) use oblivious subspace embeddings $S$ which provide
constant factor approximations for all vectors simultaneously in an
$r$-dimensional subspace of $\ell_2$, 
and enable us to say that with $Cr$ rows for a constant $C > 0$, we have
$\|SA^T\|_{p^* \to 2} = \Theta(1) \|A^T\|_{p^* \to 2}$. 
Next, (3) we use that for 
a random Gaussian matrix $G \in \mathbb{R}^{C'r \times Cr}$, for a constant $C' > 0$,
with appropriate variance, it has the property
that simultaneously for all $x \in \mathbb{R}^{Cr}$, 
$\|Gx\|_{1} = \Theta(1) \cdot \|x\|_2$. This is a special case of
Dvoretsky's theorem in functional analysis. Thus, instead of directly
approximating $\|SA^T\|_{p^* \to 2}$, we can obtain a constant factor approximation
by approximating $\|GSA^T\|_{p^* \to 1}$. This is another norm we do not
know how to directly work with, so we apply duality (1) again, 
and argue this is the same as approximating $\|AS^TG^T\|_{\infty \to p}$. A key
observation is now (4), that $\sup_{x \textrm{ s.t. }\|x\|_{\infty} = 1} \|AS^TG^Tx\|_p$ 
is realized when $x$ has each coordinate equal to $1$ or $-1$. Consequently, as
$x \in \mathbb{R}^{C'r}$, it suffices to use any sketch $T$ for the $p$-norm of
a fixed vector which fails with probability $\exp(-C'r)$, and estimate
$\|TAS^TG^Tx\|_p$ for each of the $2^{C'r}$ possible maximizers $x$, and output
the largest estimate. As there exist sketches $T$ with $O(n^{1-2/p} r \log n)$ 
rows for this
purpose, this gives us an overall sketching complexity of $O(n^{1-2/p} r^2 \log n)$. 

We defer a discussion of our lower bound techniques to \pref{sec:lowerbounds}. 

\section{Preliminaries}\label{sec:preliminaries}
In this section, we introduce the tools we use in this paper.
\begin{definition}[Total Variation Distance]\label{def:TVdistance}
Given two distributions $\D$ and $\D'$ over sample space $\Omega$ with density functions $p_{\D}$ and $p_{\mathcal{D'}}$,
the \textbf{total variation distance} is defined in two equivalent ways as follows
$
d_{TV}(\D,\D') = \frac{1}{2}\|p_{\D}-p_{\D'}\|_1 = \sup_{\mathcal{E}}
|\Pr_{x\sim\D}[\mathcal{E}]-\Pr_{x\sim\D'}[\mathcal{E}]|
$
\end{definition}

The following result bounds the total variation distance between two multivariate Gaussians.
\begin{lemma}\label{lem:gaussianTVdistance}\cite[Lemma A4]{hardt2015tight}
Let $\lambda$ be the minimum eigenvalue of PSD matrix $\Sigma$, then
$d_{TV}(\mathcal{N}(\mu,\Sigma),\mathcal{N}(\mu',\Sigma'))\leq
\frac{C}{\sqrt{\lambda}}(\|\mu-\mu'\|_2 + \|\Sigma-\Sigma'\|_F)$
for an absolute constant $C$.
\end{lemma}

We state a well known result that a Lipschitz function of a Gaussian vector is tightly concentrated around its expectation,
which is useful since $\ell_p$ norms are Lipschitz.
\begin{theorem}\label{thm:gaussianlipschitz}\cite[Theorem 2.1.12]{tao2012topics}
Let $X\sim\mathcal{N}(0,I_n)$ be a Gaussian random vector and let $f:\R^n\rightarrow\R$ be a 1-Lipschitz function. Then for
some absolute constants $C,c>0$,
$\Pr[|f(X)-\E[f(X)]|\geq\lambda]\leq C\exp(-c\lambda^2)$
Notice that this implies if $f$ is $t$-Lipschitz, then
$
\Pr[|f(X)-\E[f(X)]|\geq\lambda]\leq C\exp(-c\lambda^2/t^2)
$
\end{theorem}

It is possible to embed $\ell_2^n$ into $\ell_p^{O(n)}$ with constant distortion using a linear map when $p\in[1,2]$, and we use
the existence of such a linear map in our results.
\begin{lemma}\label{thm:dvoretzky}\cite[Theorem 2.5.1]{matouvsek2013lecture}
For all $p\in[1,2]$, there is an absolute constant $C_p$ such that for any $n$, there is a linear map
$T: \R^n \rightarrow \R^{C_pn}$ such that
$
\|T(x)\|_p = \left(1\pm\frac{1}{2}\right)\|x\|_2. 
$
An important observation is that this implies for any linear map $A:\R^n\rightarrow\R^n$, we have
$\|TA\|_{q\to p}=\left(1\pm\frac{1}{2}\right)\|A\|_{q\to 2}$.
\end{lemma}

In the lemma below we make an important observation that highlights the connection between several $p \to q$ norms. 

\begin{lemma}\label{lem:ptoqduality}
For any $p, q \geq 1$ and $d \times n$ matrix $A$, $\|A\|_{q \to p} = \|A^T\|_{p^* \to q^*}$. 
\end{lemma}
\begin{proof}
Using the notation above for dual norms, we have
\begin{align*}
\|A\|_{q \to p} &= \sup\{\|Ax\|_q: \|x\|_p \leq 1\} \\ 
&= \sup\{\sup\{y^\top Ax: \|y\|_{q*} \leq 1\}: \|x\|_p \leq 1\} \\ 
&= \sup\{\sup\{x^\top A^\top y: \|x\|_p \leq 1\}: \|y\|_{q^*} \leq 1\} \\ 
&= \sup\{\|A^\top y\|_{p^*}: \|y\|_{q^*} \leq 1\} \\
&= \|A^\top\|_{p^* \to q^*}
\end{align*}
\end{proof}
Throughout the paper, we make use of $q^*$ to refer to $\frac{q}{q-1}$ since $\ell_{\frac{q}{q-1}}$ is the dual norm of $\ell_q$.

We give a characterization of the $1\to p$ and $\infty\to p$ norm of a matrix. The proofs can be found in \pref{app:prelims}. For any $d\times n$ matrix $A$, we have
\begin{lemma}\label{lem:1toPcharac}
$\|A\|_{1\to p}=\max_{i\in[n]}\{\|A_{*,i}\|_p\}$.
\end{lemma}
\begin{lemma}\label{lem:inftytoPcharac}
$\|A\|_{\infty\to p}=\max_{x\in\{\pm 1\}^n}\|Ax\|_p$.
\end{lemma}

We introduce the machinery of $\eps$-nets, a common tool in the study of random matrices (see \cite{vershynin2010introduction})
along with some relevant lemmas and defer the proofs to the full version's Appendix.
\begin{definition}[$\eps$-net]
Let $\mathcal{X}$ be a normed space. For $S\subseteq V$, we call a set $N$ an $\eps$-net for $S$ if for all $v\in S$, there is
$v'\in N$ such that $\|v-v'\|_{\mathcal{X}}<\eps$.
\end{definition}

For a linear operator $A$, we show that to bound $\|A\|_{\mathcal{X}\to\mathcal{Y}}$, it suffices to bound
$\|Ax\|_{\mathcal{Y}}$ for $x$ taken over an $\eps$-net of the unit ball in $\mathcal{X}$.

\begin{lemma}\label{lem:epsnetball}
Let $\mathcal{X}$ and $\mathcal{Y}$ be normed spaces and let $A:\mathcal{X}\rightarrow \mathcal{Y}$ be a linear map.
Suppose $N$ is an $\eps$-net of the unit ball in $\mathcal{X}$, then
$
\|A\|_{\mathcal{X}\to\mathcal{Y}} \le \frac{1}{1-\eps}\max_{v\in N}\|Av\|_{\mathcal{Y}}.
$
\end{lemma}

We also give a way to construct `small' $\eps$-nets of unit balls.
\begin{lemma}\label{lem:epsnetsize}
There is an $\eps$-net of the unit ball $B$ in an $n$-dimensional normed space $\mathcal{X}$ with at most $\left(\frac{2+\eps}{\eps}
\right)^n$ elements. 
\end{lemma}

Another tool we use is subspace embeddings, which we define below.
\begin{definition}
An \textbf{oblivious subspace embedding family} (OSE family) is a distribution $\mathcal{S}$ over $O(m)\times n$ matrices such that for
any subspace $K\subseteq\R^n$ of dimension $m$,
$
\Pr_{S\sim\mathcal{S}}[\forall x\in K:\|Sx\|_2 = \Theta(1)\|x\|_2]\geq\frac{9}{10}.
$
\end{definition}
\begin{lemma}\label{lem:OSEfamilies}\cite{sarlos2006improved}
There exist OSE families, where the matrices have dimension $O(k)\times n$.
Note that this means for any rank-$k$ matrix $A$, a randomly drawn $S$ from such
an oblivious subspace embedding family satisfies $\|SAx\|_2 = \Theta(1)\|Ax\|_2$
simultaneously for all $x$ with probability at least $99/100$.
\end{lemma}

\section{Sketching algorithms for constant factor approximations}\label{sec:upperbounds}
\subsection{Sketches for approximating $\|A\|_{1\to p}$}\label{sec:1toPupper}

We show how to use sketches for $p$-norms of vectors to come up with sketches for the $1\to p$ norm.
\begin{lemma}\label{lem:matrixcolsketch}
Let $x$ be an arbitrary vector in $\R^n$. If $\mathcal{S}$ is a distribution over $t\times n$ sketching matrices, and
$f:\R^t\rightarrow \R$ is a function
such that
$
\Pr_{S\sim \mathcal{S}}\left[f(Sx)\in \left(\frac{1}{2}\|x\|_p,2\|x\|_p\right)\right]\geq\frac{2}{3}
$
then there is an $(O(nt\log n),2)$-sketching family $(\mathcal{S}', g)$ for the $1\to p$ norm of $n\times n$ matrices.
\end{lemma}
\begin{proof}
Proof in \pref{app:upperbounds}.
\end{proof}

Given an $n$-dimensional vector $x$, we have the following theorems from \cite{kane2010exact} and \cite{andoni2011streaming}
respectively.
\begin{theorem}[Efficient sketches for small norms]\label{thm:sketchingsmallnorm}
When $p\in[1,2]$, there is a function $f$ and a distribution over sketching matrices $\mathcal{F}$ with $O(1)$ rows such
that for $S\sim\mathcal{F}$, $f(Sx)$ is a constant factor approximation for $\|x\|_p$ with probability at least $2/3$.
\end{theorem}
\begin{theorem}[Efficient sketches for large norms]\label{thm:sketchinglargenorm}
When $p > 2$, there is a function $f$ and a distribution over sketching matrices $\mathcal{F}$ with $O(n^{1-2/p}\log n)$ rows such that
for $S\sim\mathcal{F}$, $f(Sx)$ is a constant factor approximation for $\|x\|_p$ with probability at least $2/3$.
\end{theorem}

\pref{lem:matrixcolsketch} tells us the following as a corollary to \pref{thm:sketchingsmallnorm} and
\pref{thm:sketchinglargenorm}.
\begin{theorem}
There is an $(O(n\log n),2)$-sketching family for the $1\to p$ norm when $p \in[1,2]$ and a $(O(n^{2-2/p})\log^2 n, 2)$-sketching
family for the $1\to p$ norm when $p\in(2,\infty]$.
\end{theorem}

\subsection{Sketches for approximating $\|A\|_{2\to p}$ for $p > 2$}\label{sec:2toqupper}
We give a sketching algorithm for the $2\to p$ norm of $A$, whose number of measurements depends on the rank $r$
of $d\times n$ matrix $A$.
\begin{theorem}
There is an $(O(n^{1-2/p}r^2\log n),\Theta(1))$-sketching family for the $2\to p$ norm.
\end{theorem}
\begin{proof}
Observe that $\|A\|_{2\to p}$ is equal to $\|A^T\|_{p^*\to 2}$ by \pref{lem:ptoqduality} and let
$S$ be a $Cr\times d$ matrix drawn from an oblivious subspace embedding family, which exists by \pref{lem:OSEfamilies}. From \pref{thm:dvoretzky}, let $G$ be a
$\beta r\times Cr$ map such that for all $x$, $\|GSA^Tx\|_1=\Theta(1)\|SA^Tx\|_2$.
Combining with the subspace embedding property, we get that $\|GSA^Tx\|_1=\Theta(1)\|A^Tx\|_2$ for all $x$, which is equivalent to saying
$\|GSA^T\|_{p^*\to 1}=\Theta(1)\|A\|_{2\to p}$.
Another application of \pref{lem:ptoqduality} gives us that
$\|AS^TG^T\|_{\infty\to p}=\Theta(1)\|A\|_{2\to p}$.
Since $AS^TG^T$ is $n\times \beta r$,
$\|AS^TG^T\|_{\infty\to p} = \max_{x\in\{\pm 1\}^{\beta r}}\|AS^TG^T x\|_p$.

Our final ingredient is the existence of an $O(n^{1-2/p}\log n\log(1/\delta))\times n$ sketching matrix $E$
and estimation function $f$ such that for any $x$, $\Pr[f(Ey)=\Theta(1)\|y\|_p]\geq 1-\delta$ \cite{andoni17high} when $p > 2$.
We set $\delta=2^{-2\beta r}$ and use a union bound over all $2^{\beta r}$ vectors in $\{\pm 1\}^{\beta r}$ to conclude
\begin{align*}
\Pr[\forall x\in\{\pm 1\}^{\beta r}: f(EAS^TG^Tx)=\Theta(1)\|AS^TG^Tx\|_q]&\geq 1-2^{-\beta r}\\
\Pr\left[\max_{x\in\{\pm1\}^{\beta r}}f(EAS^TG^Tx) = \Theta(1)\|AS^TG^T\|_{\infty\to q}\right]&\geq 1-2^{-\beta r}
\end{align*}
Consequently, we get a sketch that consists of $O(n^{1-2/p}r^2\log n)$ measurements to get a
$\Theta(1)$ approximation to $\|A\|_{2\to p}$ with probability at least $0.99$.
\end{proof}

\section{Sketching lower bounds for constant factor approximations}\label{sec:lowerbounds}
\subsection{Lower Bound Techniques}\label{sec:lowerboundtechnique}
The way we prove most of our lower bounds is by giving two distributions over $n\times n$ matrices, $\mathcal{D}_1$ and $\mathcal{D}_2$, where matrices drawn from the two distributions have $q\to p$ norm separated by a constant factor $\kappa$ with high probability, which means a $(k,\sqrt{\kappa})$-sketching family can distinguish between samples from the two distributions. We then show an upper bound on the variation distance between distributions of $k$-dimensional sketches of $\mathcal{D}_1$ and $\mathcal{D}_2$. We then argue that if $k$ is too small, then the total variation distance is too small to solve the distinguishing problem. We formalize this intuition in the following theorem.

\begin{theorem}\label{thm:metalowerbound}
Suppose $\mathcal{D}_1$ and $\mathcal{D}_2$ are distributions over $d\times n$ matrices such that
\begin{enumerate}[(i)]
\item $\Pr_{D\sim\mathcal{D}_1}
[\|D\|_{q\to p} < s]\geq 1-\frac{1}{n}$ and $\Pr_{D\sim\mathcal{D}_2}[\|D\|_{q\to p} > \kappa s]\geq 1-\frac{1}{n}$
\item for any linear map $L:\R^{d\times n}\rightarrow\R^k$, $d_{TV}(L(\mathcal{D}_1),L(\mathcal{D}_2))=O\left(\frac{k^a}{n^b}\right)$
\end{enumerate}
for constants $s,\kappa,a,b$, any $(k,\sqrt{\kappa})$-sketching family for the $q\to p$ norm must satisfy $k=\Omega(n^{b/a})$.
\end{theorem}
\begin{proof}
Let $\D$ be the distribution over matrices given by sampling from $\mathcal{D}_1$ with probability $\half$ and drawing from
$\mathcal{D}_2$ with probability $\half$.
We shall fix a sketching operator $L:\R^{d\times n}\rightarrow\R^k$ and consider $A$ drawn from a distribution $\mathcal{D}$. Suppose $f(L(A))$ lies in $(1/\sqrt{\kappa},\sqrt{\kappa})\|A\|_{q\to p}$ with probability at least $5/6$. It suffices to show that
$k$ must be $\Omega(n^{b/a})$ since the theorem statement then follows from Yao's minimax principle. We must have
\[\Pr_{A\sim\mathcal{D}_1}\left[f(L(A))\in \left(\frac{1}{\sqrt{\kappa}},\sqrt{\kappa}\right)\|A\|_{q\to p}\right]
\geq\frac{2}{3},~
\Pr_{A\sim\mathcal{D}_2}\left[f(L(A))\in \left(\frac{1}{\sqrt{\kappa}},\sqrt{\kappa}\right)\|A\|_{q\to p}\right]
\geq\frac{2}{3}\]

Thus, we have an algorithm that correctly distinguishes with probability at least $\frac{3}{5}$ if $A$ was drawn from
$\mathcal{D}_1$ or $\mathcal{D}_2$ by checking if $f(L(A))$ is greater than or less than $\sqrt{\kappa} s$.

The existence of this distinguishing algorithm means the total variation distance between the distributions of $L(D_1)$ and
$L(D_2)$ is at least $\frac{1}{5}$. From the theorem's hypothesis, we know of a constant $C$ such that $\frac{Ck^{a}}{n^b}\geq
\frac{1}{5}$, which gives us the desired upper bound.
\end{proof}

We also show an upper bound on the variation distance of sketches for two distributions that we use throughout this paper. Define
$\mathcal{G}_{1,d\times n}$
as the distribution over $d\times n$ Gaussian matrices and $\mathcal{G}_{2,d\times n}[\alpha]$ as the distribution given by drawing a
Gaussian matrix and adding $\alpha u$, where $u$ is a $d$-dimensional Gaussian vector to a random column. We write $\mathcal{G}_i$
instead of $\mathcal{G}_{i,d\times n}$ when the dimensions of the random matrix are evident from context.
\begin{lemma}\label{lem:mixtureTVdistance}
Let $L$ be a linear sketch from $\R^{d\times n}\rightarrow \R^k$ and let $\mathcal{H}_i$ be the distribution of $L(x)$ where $x$ is
drawn from $\mathcal{G}_i$. Then $d_{TV}(\mathcal{H}_1,\mathcal{H}_2)\leq\frac{C\alpha^2 k}{n}$ for an absolute constant $C$.
\end{lemma}
\begin{proof} We can think of $L$ as a $k\times nd$ matrix that acts on a sample from $\mathcal{G}_1$ or $\mathcal{G}_2$
as though it were an $nd$-dimensional vector.
Without loss of generality, we can assume that the rows of $L$ are orthonormal, since one can always perform a change of basis in post-processing. 
Thus, the distribution $\mathcal{H}_1$
is the same as $\mathcal{N}(0,I_k)$. For fixed $i$ and $G$ a $d\times n$ matrix of unit Gaussians, the distribution of
$L(G+\alpha ue_i^T)$ is Gaussian with covariance $\E[L(G+\alpha ue_i^T)L(G+\alpha ue_i^T)^T]$,
equal to $I+\alpha^2 L_{B_i}L_{B_i}^T$ where $L_{B_i}$ is the submatrix given by columns of $L$ indexed $(i-1)d+1,(i-1)d+2,\ldots,id$.
Let $\mathcal{H}_{2,i}$ be $\mathcal{N}(0,I+\alpha^2L_{B_i}L_{B_i}^T)$.
$\mathcal{H}_2$ is the distribution of picking a random $i$ and drawing a matrix from $\mathcal{N}(0,I+L_{B_i}L_{B_i}^T)$.

We now analyze the total variation distance between $\mathcal{H}_1$ and $\mathcal{H}_2$ and get the desired bound from a chain of inequalities.
\begin{align*}
d_{TV}(\mathcal{H}_1,\mathcal{H}_2) &= \frac{1}{2}\int_{x\in\R^k} |p_{\mathcal{H}_1}(x)-p_{\mathcal{H}_2}(x)|dx\\
&\le \frac{1}{2}\int_{x\in\R^k}\left|\sum_{i=1}^n \frac{1}{n}p_{\mathcal{H}_1}(x) - \frac{1}{n}p_{\mathcal{H}_{2,i}}(x)\right|dx\\
&\le \frac{1}{n}\sum_{i=1}^n\frac{1}{2}\int_{x\in\R^k}\left|p_{\mathcal{H}_1}(x)- p_{\mathcal{H}_{2,i}}(x)\right|dx\\
&\le \frac{1}{n}\sum_{i=1}^n d_{TV}(\mathcal{N}(0,I_k),\mathcal{H}_{2,i})\\
&\le \frac{1}{n}\sum_{i=1}^n C\alpha^2\|L_{B_i}L_{B_i}^T\|_F &\text{[from \pref{lem:gaussianTVdistance}]}\\
&\le \frac{1}{n}\sum_{i=1}^n C\alpha^2\|L_{B_i}\|_F^2\\
&\le \frac{C\alpha^2}{n}\|L\|_F^2 = \frac{C\alpha^2 k}{n}
\end{align*}
\end{proof}

\subsection{Lower bounds for approximating $\|A\|_{1\to p}$ for $1\leq p\leq 2$}\label{sec:1toplower}
We follow the lower bound template given in \pref{sec:lowerboundtechnique}.
\begin{lemma}\label{lem:1topseparation}
For any $\kappa$, there exist values $s_p$ such that with probability at least $1-1/n$, $\|G_1\|_{1\to p}\le s_p$ and $\|G_2\|_{1\to p}\ge \kappa s_p$, for $1\leq p\leq 2$, and $G_1\sim\mathcal{G}_1$ and $G_2\sim\mathcal{G}_2[\kappa]$.
\end{lemma}
\begin{proof}
Recall that from \pref{sec:1toPupper}, we know that $\|A\|_{1\to p}=\max_{i\in[n]}\|A_{*,i}\|_p$ which means that it
suffices to give bounds on the maximum $\ell_p$ norm across columns of $G_1$ and $G_2$ respectively.

The $\ell_p$ norm is $\zeta_p$-Lipschitz, where $\zeta_p$ is equal to $n^{1/p-1/2}$ in the regime $1\leq p \leq 2$.
For a given vector of standard Gaussians $g$, the probability that $\|g\|_p$ deviates from $\E\left[\|g\|_p\right]$ by
more than $\beta\zeta_p\sqrt{\log n}$ is at most $C'e^{-c\beta^2\log n}$ from \pref{thm:gaussianlipschitz} where
$C'$ is the constant $C$ from the theorem, which for large enough choice of $\beta$ can be made smaller than $1/n^2$.
By a union bound over all columns, the probability that $\|G_1\|_{1\to p}$ exceeds $\E[\|g\|_p]+\beta\zeta_p\sqrt{\log n}$ is at most
$1/n$.
On the other hand, consider the perturbed column vector of $G_2$, which we denote $g'$. The probability that $\|g'\|_2$ is smaller
than
$\E[\|g'\|_p]-\beta\sqrt{1+\kappa^2}\zeta_p\sqrt{\log n} = \sqrt{1+\kappa^2}(\E[\|g\|_p]-\beta\zeta_p\sqrt{\log n})$
is at most $1/n^2$ by appropriate choice of $\beta$ and \pref{thm:gaussianlipschitz}, from which a lower bound
on $\|G_2\|_{1\to p}$ that holds with probability at least $1-\frac{1}{n^2}$ immediately follows.

Since $\E[\|g\|_p]$ is $\Theta(n^{1/p})$ and the deviations from expectations in upper bounds on $\|G_1\|_{1\to p}$ and lower bounds on
$\|G_2\|_{1\to p}$ are asymptotically less than the expectations.
\end{proof}

The desired theorem is immediate from \pref{lem:1topseparation}, \pref{lem:mixtureTVdistance}, and
\pref{thm:metalowerbound} using $\D_1=\mathcal{G}_{1,n\times n}$, and $\mathcal{D}_2=\mathcal{G}_2[\kappa]$.
\begin{theorem}\label{thm:1to12Lower}
Suppose $p\in[1,2]$ and $(\mathcal{S},f)$ is a $(k,\sqrt{\kappa})$-sketching family for the $1\to p$ norm
where $\kappa$ is some constant, then $k=\Omega(n)$.
\end{theorem}

\subsection{Lower bound for approximating $\|A\|_{1\to p}$ for $p > 2$}\label{sec:1toInftylower}
We follow the lower bound template given in \pref{sec:lowerboundtechnique}.

Denote $\E[\|g\|_p]$ as $\eta_p$.
Let $\mathcal{G}_1$ be the distribution over $n\times n$ matrices given by i.i.d. Gaussians, and
$\mathcal{G}_2[\alpha, \eta_p]$ be the distribution over $n\times n$ matrices given by taking a Gaussian matrix and adding
$\alpha\eta_p$ to a random entry.

Since the proofs are very similar to those in \pref{sec:lowerboundtechnique} and \pref{sec:1toplower}. We defer them to
\pref{app:1toInftylower}.

\begin{lemma}\label{lem:1toinftyseparation}
For any $\kappa$, there exists $s_p$ such that with probability at least $1-\frac{1}{n}$, $\|G_1\|_{1\to p} \le s_p$ and $\|G_2\|_{1\to p}\ge \kappa s_p$, such that $G_1\sim\mathcal{G}_1$ and $G_2\sim\mathcal{G}_2[C\kappa, \eta_p]$ for some absolute constant $C$ and $p > 2$.
\end{lemma}

\begin{lemma}\label{lem:variationdist1toinfty}
Let $L$ be a linear sketch from $\R^{n\times n}\rightarrow \R^k$ and let $\D_i$ be the distribution of $L(x)$ where $x$ is
drawn from $\mathcal{G}_i$. Then $d_{TV}(\D_1,\D_2)\leq \frac{C'\alpha \eta_p \sqrt{k} }{n}$ for an absolute constant $C'$.
\end{lemma}

The theorem below immediately follows from \pref{lem:1toinftyseparation}, \pref{lem:variationdist1toinfty} and
\pref{thm:metalowerbound} using $\mathcal{D}_1=\mathcal{G}_1$ and $\mathcal{D}_2=\mathcal{G}_2[C\kappa, \eta_p]$.
\begin{theorem}\label{thm:1toinftylower}
Suppose $(\mathcal{S},f)$ is a $(k,\kappa)$-approximate sketching family for the $1\to p$ norm for $p > 2$ and some constant $\kappa$,
then $k=\Omega\left( \frac{n^2}{\eta_p^2} \right)$. In particular, using the fact that $\eta_p$ is $\Theta(n^{1/p})$ for $p<\infty$ and
$\Theta(\sqrt{\log n})$ when $p=\infty$ gives $k = \Omega \left(n^{2 - \frac{2}{p}} \right)$ when $p<\infty$ and
$k = \Omega\left(\frac{n^2}{\log n} \right)$ when $p=\infty$.
\end{theorem}

\subsection{Lower bound for approximating $\|A\|_{q\to p}$ when $q\geq 2$ and $p\leq 2$}\label{sec:qlarge_to_psmall_lower}
We use the known lower bound of $\Omega(n^2)$ for sketching the $2\to 2$ norm from \cite{li2016tight} to deduce a lower bound
on sketching the $q\to p$ norm for $q\geq 2$ and $p\leq 2$.
\begin{theorem}\label{qlarge_to_psmall_lower}
Suppose $q\geq 2$ and $p\leq 2$, and if $(\mathcal{S},f)$ is a $(k(n),\gamma)$-approximate sketching family for the $q\to p$
norm where $\gamma$ is some constant, then $k(n)=\Omega(n^2)$.
\end{theorem}

\begin{proof}
We prove this by showing that if the hypothesis of the theorem statement holds, then the $2\to 2$ norm can be sketched in $O(k)$ measurements.

Given an $n\times n$ matrix $A$ for which we want to sketch the $2\to 2$ norm, note that by \pref{thm:dvoretzky}
there is a $Cn\times n$ matrix $L_1$ such that $\|L_1A\|_{2\to q^*} = (\frac{1}{\beta},\beta)\|A\|_{2\to 2}$ for a constant
$\beta$, and by \pref{lem:ptoqduality} $\|L_1A\|_{2\to q^*}=\|A^TL_1^T\|_{q\to 2}$, and another application of
\pref{thm:dvoretzky} gives us another $Cn\times n$ matrix $L_2$ for which $\|L_2A^TL_1^T\|_{q\to p} = (\frac{1}{\beta},
\beta)\|A^TL_1^T\|_{q\to 2}$.
Note that this means $\|L_2A^TL_1^T\|_{q\to p} = \left(\frac{1}{\beta^2},\beta^2\right)\|A\|_{2\to 2}$, 
so we can sketch $A$ by drawing a random $L$ from $\mathcal{D}$ and storing $L(L_2A^TL_1^T)$, which uses 
$k(Cn)$ measurements and serves as a sketch from which $f$ can be used to estimate $\|A\|_{2\to 2}$ within a constant factor,
which means from \cite{li2016tight}, $k(Cn)$ must be $\Omega(n^2)$, which means $k(n)=\Omega(n^2/C^2)=\Omega(n^2)$.
\end{proof}

\subsection{Lower bounds for approximating $\|A\|_{q\to p}$ for $p,q\leq 2$ and $p,q\geq 2$}\label{sec:ptoq_oneside}
In this section, we show a lower bound on the sketching complexity of $\|A\|_{q\to p}$ where $A$ is a rank $r$ matrix, when both $p$ and $q$ are at most $2$. A corresponding lower bound for when $p$ and $q$ are at least 2 follows from
\pref{lem:ptoqduality}. We achieve this by first showing a lower bound on the sketching
complexity of $\|A\|_{2\to q}$ and then use Dvoretzky's theorem along with the relation between the $q\to p$ norm and the $p^*\to q^*$
norm to deduce the result.

We show a lower bound for sketching the $2\to q$ norm using the template from \pref{sec:lowerboundtechnique}.
We use distributions $\mathcal{D}_1 = \mathcal{G}_{1,r\times n}$ and $\mathcal{D}_2[\alpha]=\mathcal{G}_{2,r\times n}\left
[\alpha\frac{d}{\sqrt{r}}\right]$, as defined in \pref{sec:lowerboundtechnique} where $d$ is $\max\{n^{1/q},\sqrt{r}\}$.
\begin{lemma}\label{lem:2toqsep}
There exist values $s_q$ and $t_q$ such that with high probability, $\|G_1\|_{2\to q}\le s_q$ and $\|G_2\|_{2\to q}\ge
C\alpha s_q$ for some absolute constant $C$, for $q > 2$, and $G_1\sim\mathcal{D}_1$ and $G_2\sim\mathcal{D}_2[\alpha]$.
\end{lemma}
\begin{proof}
Let $N$ be a $1/3$-net of the Euclidean ball in $\R^r$ with $7^r$ elements, which exists by \pref{lem:epsnetsize}. For a fixed
$x\in N$, $G_1 x$ is distributed as an $n$-dimensional vector with independent Gaussians, whose $q$-norm is at most $\beta_1n^{1/q}$ for
some constant $\beta_1$ in expectation and exceeds $\beta_1n^{1/q} + \beta_2\sqrt{r}$ with probability at most $\frac{1}{8^r}$ for
appropriate constant $\beta_2$, which follows from the $q$-norm being 1-Lipschitz and \pref{thm:gaussianlipschitz}.
A union bound over all $x\in N$ implies that with probability at least $1-(7/8)^r$, $\forall x\in N:\|G_1x\|_q\leq\beta_1n^{1/q}+\beta_2\sqrt{r}$.

Then by applying \pref{lem:epsnetball}, we conclude that with probability at least $1-(7/8)^r$,
$\|G_1\|_{2\to q}\leq\frac{3}{2}(\beta_1n^{1/q}+\beta_2\sqrt{r})\leq \frac{3}{2}(\beta_1+\beta_2)d$.
On the other hand, the perturbed row of $G_2$, called $g'$ is distributed as $\sqrt{1+\alpha^2\frac{d^2}{r}}g$ for a vector of
i.i.d. Gaussians $g$. If we take the unit vector $u$ in the direction of $g'$, then
the entry of $G_2u$ corresponding to the perturbed row is concentrated around $\sqrt{1+\alpha^2\frac{d^2}{r}}\|g\|_2=
\sqrt{r+\alpha^2d^2}$, which means
$\|G_2\|_{2\to q}\geq(1-o(1))\sqrt{r+\alpha^2d^2}\geq 0.9\alpha d$
with high probability.
\end{proof}

The theorem below immediately follows from \pref{lem:2toqsep}, \pref{lem:mixtureTVdistance} and \pref{thm:metalowerbound}.
\begin{theorem}\label{thm:2toqlower}
Suppose $q\geq 2$ and $(\mathcal{S},f)$ is a $(k,\gamma)$-sketching family for the
$2\to q$ norm of rank $r$ matrices for some constant $\gamma$. Then $k=\Omega(nr/d^2)$.
\end{theorem}

\begin{theorem}\label{thm:ptoqless2lower}
Suppose $p,q\leq 2$ and $(\mathcal{S},f)$ is a $(k,\gamma)$-sketching family for the $q\to p$ norm of rank $r$
matrices for some constant $\gamma$. Then $k=\Omega(nr/d^2)$ where $d=\max\{\sqrt{r},n^{1/q^*}\}$.
\end{theorem}
\begin{proof}
For a matrix $A$, from \pref{lem:ptoqduality} we have that $\|A\|_{2\to q^*} = \|A^T\|_{q\to 2}$, and from
\pref{thm:dvoretzky}, we know there is a $Cr\times r$ matrix $L_1$ such that $\|L_1A^T\|_{q^*\to p}=\Theta(1)\|A\|_{2\to q^*}$.
We can use $(\mathcal{S},f)$ to sketch $L_1A^T$ to obtain an $(O(k),\Theta(1))$-sketching family for the
$2\to q^*$ norm, whose lower bound from \pref{thm:2toqlower} gives us the desired lower bound.
\end{proof}

\subsection{Lower bounds for approximating $\|A\|_{q\to p}$ for $1\leq q\leq 2$ and $p\geq 2$}\label{sec:qsmall_plarge}
We prove the desired lower bound using the template from \pref{sec:lowerboundtechnique}.
Let $\mathcal{D}_1$ be a distribution over $n\times n$ matrices where diagonal entries are Gaussians and off-diagonal
entries are 0 and let $\mathcal{D}_2[\alpha]$ be a distribution over $n\times n$ matrices where a matrix is drawn from $\mathcal{D}_1$
and $\alpha\sqrt{\log n}$ is added to a random diagonal entry.

\begin{lemma}\label{lem:qsmallplargesep}
There exists values $s_{p,q}$, $t_{p,q}$ and $\alpha$ such that with probability at least $1-1/n$, $\|G_1\|_{q\to p}\leq s_{p,q}$
and $\|G_2\|_{q\to p}\geq \kappa s_{p,q}$ for some desired constant factor $\kappa$ separation, such
that $G_1\sim\mathcal{D}_1$ and $G_2\sim\mathcal{D}_2[\alpha]$.
\end{lemma}
We give the proof of \pref{lem:qsmallplargesep} in \pref{app:qsmall_plarge}.

Without loss of generality, we can assume that any sketch of $G_1$ and $G_2$ acts on $\mathrm{diag}(G_1)$ and $\mathrm{diag}(G_2)$
respectively. \pref{lem:variationdist1toinfty} gives an upper bound of $O(\sqrt{k\log n}/\sqrt{n})$ on the variation distance
between $k$-dimensional sketches of these distributions. Thus, from the variation distance bound, \pref{lem:qsmallplargesep}
and \pref{thm:metalowerbound}, the desired theorem follows.

\begin{theorem}\label{thm:qsmall_plarge_lowerbound}
Suppose $q\geq 2$ and $(\mathcal{S},f)$ is a $(k,\gamma)$-sketching family for the $q\to p$ norm of rank $r$ matrices for some constant
$\gamma$, then $k=\Omega(n/\log n)$.
\end{theorem}

\section{Sketching with large approximation factors}\label{sec:largeapprox}

While our results primarily involve constant factor approximations, we give several preliminary results studying large approximation factors for sketching the important cases of the $2 \to q$ norm and $[1, \infty] \to [1, 2]$ norms. Our goal is, given an approximation factor $\alpha(n)$, to give upper and lower bounds on $k$ for a $(k, \alpha(n))$-sketching family for the respective norms. As a shorthand, we will refer to $\alpha(n)$ as $\alpha$.
\subsection{Sketching upper bounds for large approximations of $\|A\|_{2 \to q}$}
It is sufficient to give a $(k, \alpha)$-sketching family for the $\infty \to q$ norm. To see why, given an input matrix $A \in \R^{n \times n}$, by \pref{lem:ptoqduality} we have that $\|A\|_{2\to q}=\|A^T\|_{q^*\to 2}$. Using \pref{thm:dvoretzky}, there is a linear map such that this is equal within a constant factor of $\|GA^T\|_{q^*\to 1}=\|AG^T\|_{\infty \to q}$.
\begin{theorem}\label{thm:largeapprox2to4}
Given a matrix $A \in \R^{n \times n}$, there exists a $(O(\frac{n^2}{\alpha}), \alpha)$-sketching family given by $(\mathcal{S}, f)$ for the $\infty \to q$ norm.
\end{theorem}
\begin{proof}
Let  $B \in \Z^+$ be some positive integer to be chosen later. Let the columns of our sketch matrix $S$ be indexed by sets given by $\{B_i\}_{i =1}^{n/B}$ such that $B_i = ((i-1)B, iB]$. For each column $v_{B_i}$, we define i.i.d random variables $\{\sigma_{ij}\}_{j = 1}^B$ such that $\sigma_{ij} = 1$ with probability $\frac{1}{2}$ and $-1$ with probability $\frac{1}{2}$. Let the column $v_{B_i}$ be as follows: 
$$v_{B_i}[j] = \begin{cases} \sigma_{ij} & \text{for } j \in [(i-1)B, iB] \\ 
							0        & \text{o/w} 
\end{cases}$$
We define our linear map $L(A)$ to be $L(A) = AS$. Our function $f : \R^{n/B} \rightarrow \R$ simply optimizes over $\{-1, 1\}^{n/B}$ and outputs $\|AS\|_{\infty \to q}$.

Since all $\sigma_{ij} \in \{-1, 1\}$ we have that $f(L(A)) \leq \|A\|_{\infty \to q}$ since $Sx$ for $x \in \{-1, 1\}^{n/B}$ has the property that $Sx \in \{-1, 1\}^n$.

We now show a lower bound on $f(L(A))$. To do so, we let $T_i$ denote the column indices of $A$ such that the index is column $i$ in its respective block. We then notice that there exists $i \in [n/B]$ such that $\|A_{*, T_i}\|_{\infty \to q} \geq \frac{B}{n}\|A\|_{\infty \to q}$. We get this by applying the triangle inequality $\|A\|_{\infty \to q} \leq \sum_{i = 1}^{n/B} \|A_{*, T_i}\|_{\infty \to q}$.

Let $i^*$ be the index that realizes this $n/B$-approximation to $\|A\|_{\infty \to q}$ and let $\{s_1\}_{i = 1}^{n/B}$ be the assignment of signs that realizes the $\infty \to q$ norm of $A_{*, T_{i^*}}$.
\begin{align*}
f(L(A)) \geq \| \sum_{i = 1}^B \sum_{j = 1}^{n/B} s_jA_{*, B_j[i]}  \|_q \geq \| \underbrace{\sum_{j = 1}^{n/B} s_jA_{*, B_j[i^*]}}_y +  \underbrace{\sum_{i \neq i^*}^B \sum_{j = 1}^{n/B} s_jA_{*, B_j[i]}}_{z}  \|_q
\end{align*}
Notice that $z$ is symmetric around the origin and hence we get that $\|y + z + y - z\|_q \leq \frac{\|y + z\|_q + \|y - z\|_q}{2}$ which implies that $f(L(A)) \geq \|y+z\|_q \geq \Theta(1)\|y\|_q \geq \frac{n}{B}\|A\|_{\infty \to q}$ with probability at least $\frac{1}{2}$. Thus, we get an $O \left(\frac{n^2}{\alpha}\right)$ space sketch that gives us an $\alpha$-approximation by setting $B=n/\alpha$.
\end{proof}

\subsection{Sketching upper bounds for large approximations of $\|A\|_{q \to p}$ for $q \in [1, \infty]$ and $p \in [1, 2]$}

We give a description of our sketch followed by the approximation factor. Towards the end of defining our sketch, let $B \in \Z^+$ be some positive integer to be chosen later. Let the rows of our sketch matrix $S$ be indexed by sets given by $\{B_i\}_{i =1}^{n/B}$ such that $B_i = ((i-1)B, iB]$. For each row $v_{B_i}$, we define i.i.d random variables $\{\sigma_{ij}\}_{j = 1}^B$ such that $\sigma_{ij} = 1$ with probability $\frac{1}{2}$ and $-1$ with probability $\frac{1}{2}$. Let the row $v_{B_i}$ be as follows: 
$$v_{B_i}[j] = \begin{cases} \sigma_{ij} & \text{for } j \in [(i-1)B, iB] \\ 
							0        & \text{o/w} 
\end{cases}$$
Our algorithm simply outputs $\|SA\|_{q \to p}$. The proof of the theorem below can be found in \pref{sec:generalapproxappendix}.
\begin{theorem}
Given a matrix $A \in \R^{n \times n}$, there exists an $(\tilde{O}(\frac{n^2}{\alpha^2}), \alpha)$-sketching family given by $(\mathcal{S}, f)$ for the $q \to p$ norm for $p \in [1, 2]$.
\end{theorem}

\section{Further Directions}

One interesting direction is to study the low-rank approximation problem with respect to the $q \to p$ norm. An important open question in the literature is to find input sparsity time low rank approximation algorithms with respect to the $2\to2$ norm, and a natural step might be to try this problem with for $q\to p$ norms for certain $q$ and $p$.

Another interesting problem would be to investigate algorithms for approximate nearest neighbors with respect to the $q \to p$ norm, in light of a question posed by \cite{andoni2017approximate} about what metric spaces admit efficient approximate nearest neighbor algorithms, with matrix norms mentioned as an object of interest.

\bibliographystyle{alpha}
\bibliography{matrix_norm_sketch}

\appendix
\section{Proofs from \pref{sec:preliminaries}}\label{app:prelims}

\begin{proof}[Proof of \pref{lem:1toPcharac}]
For any $x$ that is unit according to $\ell_1$,
\begin{align*}
\|Ax\|_p &= \|A_{*,1}x_1 + A_{*,2}x_2 + \ldots + A_{*,n}x_n\|_p\\
&\leq \|A_{*,1}\|_p|x_1| + \|A_{*,2}\|_p|x_2| + \ldots + \|A_{*,n}\|_p|x_n|\leq \max_{i\in[n]}\{\|A_{*,i}\|_p\}
\end{align*}
where the last inequality is because $|x_i|$ give a convex combination and is achieved for $x=e_{i^*}$ where
$i^* = \arg\max_i\{\|A_{*,i}\|_p\}$.
\end{proof}

\begin{proof}[Proof of \pref{lem:inftytoPcharac}]
For any $x$ such that there is a coordinate $x_j$ that is strictly between 1 or $-1$, let $\eps$ be $\min\{1-x_j,x_j+1\}$, consider
\begin{align*}
\|Ax\|_p &= \|A_{*,j}x_j+\sum_{i\neq j} A_{*,i}x_i\|_p\\
&\leq \left(\frac{1+x_j}{2}\right)\|A_{*,j}+\sum_{i\neq j} A_{*,i}x_i\|_p +
\left(\frac{1-x_j}{2}\right)\|-A_{*,j}+\sum_{i\neq j} A_{*,i}x_i\|_p
\end{align*}
where the inequality is due to the triangle inequality. Since $\|Ax\|_p$ is at most a convex combination of the $p$-norms
after replacing $x_j$ with $1$ or $-1$, we can make $x_j$ one of $1$ or $-1$ without decreasing the $p$-norm.
\end{proof}

\begin{proof}[Proof of \pref{lem:epsnetball}]
Pick $x^*$ on the unit ball such that $\|Ax^*\|_{\mathcal{Y}}=\|A\|_{\mathcal{X}\to\mathcal{Y}}$. There is $x\in N$
such that $\|x^*-x\|_{\mathcal{X}}<\eps$, which means
\[\|A(x^*-x)\|_{\mathcal{Y}}\leq\|A\|_{\mathcal{X}\to\mathcal{Y}}
\|x-x^*\|_{\mathcal{X}}<\eps\|A\|_{\mathcal{X}\to\mathcal{Y}}\]
On the other hand,
\[\|A(x^*-x)\|_{\mathcal{Y}} \geq \|Ax^*\|_{\mathcal{Y}}-\|Ax\|_{\mathcal{Y}}\geq\|A\|_{\mathcal{X}\to\mathcal{Y}}
-\|Ax\|_{\mathcal{Y}}\]
and hence
\begin{align*}
\|A\|_{\mathcal{X}\to\mathcal{Y}}-\|Ax\|_{\mathcal{Y}} &< \eps\|A\|_{\mathcal{X}\to\mathcal{Y}}\\
\|A\|_{\mathcal{X}\to\mathcal{Y}} &< \frac{\|Ax\|_{\mathcal{Y}}}{1-\eps} \leq
\frac{1}{1-\eps}\max_{x\in N}\|Ax\|_{\mathcal{Y}}
\end{align*}
\end{proof}

\begin{proof}[Proof of \pref{lem:epsnetsize}]
For $x$ in a normed space $\mathcal{X}$, we use the notation $B_x(r)$ to denote $\{y:\|x-y\|_{\mathcal{X}}<r\}$, the ball of radius
$r$ around $x$.

Start with an empty set $N$ and while there is a point $x$ in the unit ball $B$ that has distance at least $\eps$ to every element in
$N$, pick $x$ and add it to $N$. This process terminates when every $x\in B$ has distance less than $\eps$ to some element in $N$, 
thereby terminating with $N$ as an $\eps$-net. We claim that the size of $N$ meets the desired bound.

By construction, any $y$ and $y'$ in $N$ are at least $\eps$ apart, which means $\mathcal{B}=\{B_x(\eps/2):x\in N\}$ is a collection
of disjoint sets and note that
\[\bigcup_{S\in\mathcal{B}}S\subseteq B_0(1+\eps/2)\]

By disjointness
\[
\mathrm{Vol}\left(\bigcup_{S\in\mathcal{B}}S\right) = \sum_{S\in\mathcal{B}} \mathrm{Vol}(S) = |N|\mathrm{Vol}(B_0(\eps/2))
\]
where $\mathrm{Vol}(S)$ is the volume of $S$ according to the Lebesgue measure.

And thus, we obtain
\begin{align*}
|N| &= \frac{\mathrm{Vol}\left(\bigcup_{S\in\mathcal{B}}S\right)}{\mathrm{Vol}(B_0(\eps/2))}\\
&\leq\frac{\mathrm{Vol}(B_0(1+\eps/2))}{\mathrm{Vol}(B_0(\eps/2))}\\
&= \left(\frac{1+\eps/2}{\eps/2}\right)^n\\
&= \left(\frac{2+\eps}{\eps}\right)^n
\end{align*}
which concludes the proof.

\end{proof}

\section{Missing proofs from \pref{sec:upperbounds}}\label{app:upperbounds}

\begin{proof}[Proof of \pref{lem:matrixcolsketch}]
Draw $c\log n$ matrices $S_1,S_2,\ldots, S_{c\log n}$ from $\D$ independently where $c$ is a constant to be determined
later. We define
\begin{align*}
S &:= \begin{bmatrix}
S_1\\
S_2\\
\vdots\\
S_{c\log n}
\end{bmatrix}\\
g(Sx) &:= \mathsf{median}\{f(S_1x),f(S_2x),\ldots,f(S_{c\log n}x)\}
\end{align*}
Let's analyze the probability that $g(Sx)$ falls outside $L_x=\left(\frac{1}{2}\|x\|_p,2\|x\|_p\right)$. In order for that
to happen, more than half of $f(S_1x),\ldots,f(S_{c\log n}x)$ must lie outside $L_x$, and this happens to
each $f(S_ix)$ with probability at most $\frac{1}{3}$. Using Hoeffding's inequality, we know
\[
\Pr[g(Sx)\notin L]\leq 2\exp\left(-\frac{c\log n}{72}\right)
\]
which for appropriate choice of $c$ can be bounded by $\frac{1}{n^2}$.

For a matrix $A$ with $n$ columns, a union bound tells us that for all $i$, $g(SA_{*,i})$ falls in $L_{A_{*,i}}$
with probability at least $1-\frac{1}{n}$. Combined with \pref{lem:1toPcharac}, it follows that $h(SA):=\max_i g(SA_{*,i})$
is a $2$-approximation to $\|A\|_{1\to p}$ with probability at least $1-\frac{1}{n}$.
\end{proof}

\section{Missing Proofs from \pref{sec:lowerbounds}}\label{app:lowerbounds}
\subsection{Missing Proofs from \pref{sec:1toInftylower}}\label{app:1toInftylower}
\begin{proof}[Proof of \pref{lem:1toinftyseparation}]
We denote $C\kappa$ as $\alpha$ and set the exact value of $\alpha$ in the end of the proof.
For a fixed pair $i,j$ let us denote the perturbation term $\alpha\eta_p e_ie_j^\top$ as $E_{ij}$.
Recall that from \pref{sec:1toPupper}, we know that $\|A\|_{1\to p}=\max_{i\in[n]}\|A_{*,i}\|_p$ which means that it
suffices to give bounds on the maximum $\ell_p$ norm across columns of $G_1$ and $G_2$ respectively.

Since the $\ell_p$ norm is $1$-Lipschitz for any $p \geq 2$, we can apply \pref{thm:gaussianlipschitz} to show concentration around the expectation for $\|G_{*, i}\|_p$ for any column $i$ of a matrix $G$ of i.i.d Gaussian entries. Hence we have that for any column $i$, and some positive constant $\lambda$
\begin{align*}
\Pr \left[\|G_{*, i}\|_p \geq \lambda \E[\|G_{*, i}\|_p] \right] \leq C \exp(-c\lambda^2 \E[\|G_{*, i}\|_p]^2)
\end{align*}

Letting $g$ be an $n$-dimensional vector of i.i.d Gaussians, since we know $\E[\|g\|_p]=\Omega(\sqrt{\log n})$, there exists
appropriate constant $\beta$ such that for any column $i$ of $G_1$ we have that $\|(G_1)_{*, i}\|_p$ is less than
$\beta\E[\|g\|_p]$ with probability at least $1 - \frac{1}{n^2}$. By a union bound over all columns, the
probability that $\|G_1\|_{1\to p} \leq \beta\E[\|g\|_p]$ is at least $1 - \frac{1}{n}$.

For a matrix $G_2=G+E_{ij}$ drawn from $\mathcal{G}_2[\alpha,\eta_p]$, we know that the perturbed column $j$ has norm at least
$\alpha\eta_p-\|G_{*,i}\|_p$, which satisfies $(\alpha - \beta)\E[\|g\|_p] \leq \|G_2\|_{1\to p}$.
Setting $\alpha \geq (\kappa+1) \beta$ gives us the desired result. 
\end{proof}

\begin{proof}[Proof of \pref{lem:variationdist1toinfty}]
Recall perturbation term $\alpha\eta_p e_ie_j^\top$ was referred to as $E_{ij}$.
Just as in \pref{lem:mixtureTVdistance}, we can think of $L$ as a $k\times n^2$ matrix that acts on a sample from $\mathcal{G}_1$ or $\mathcal{G}_2[\alpha]$ as though it were an $n^2$-dimensional vector. Without loss of generality, we can assume that the rows of $L$ are orthonormal, since as before we can always perform a change of basis in post-processing. Thus, the distribution $\D_1$ is the same as $\mathcal{N}(0,I_k)$. For fixed $i,j$, the distribution of $L(G+E_{ij})$ is Gaussian with mean vector $L(E_{ij})$ (the $ij^{\text{th}}$ column of the $k \times n^2$ matrix $L$ scaled by $\alpha\eta_p$) and covariance $I_k$ because of the following.

\begin{align*}
\Cov(L(G + E_{ij})) &= \E{ \big( L(G + E_{ij}) - \E{L(G + E_{ij})} \big)^\top \big(L(G + E_{ij}) - \E{L(G + E_{ij})}\big) } \\ 
&= \E{ \big( L(G) - \E{L(G)} \big)^\top \big(L(G) - \E{L(G)}\big) } \\
&= \Cov_{G \sim \mathcal{N}(0, I_n)}(G) = I_k
\end{align*}

Thus, $\D_2$ is the distribution of picking a random $i,j$ and drawing a matrix from $\mathcal{N}(L(E_{ij}),I_k)$.

We now analyze the total variation distance between $\D_1$ and $\D_2$ and get the desired bound from a chain of inequalities.

\begin{align*}
d_{TV}(\D_1,\D_2) &= \frac{1}{2}\int_{x\in\R^k} |p_{\D_1}(x)-p_{\D_2}(x)|dx\\
&= \frac{1}{2}\int_{x\in\R^k}\left|\sum_{i,j} \frac{1}{n^2} p_{\D_1}(x) - \frac{1}{n^2} p_{\mathcal{N}(L(E_{ij}),I_k)}(x)\right|dx\\
&\le \frac{1}{n^2}\sum_{i,j}\frac{1}{2}\int_{x\in\R^k}\left|p_{\D_1}(x)-
p_{\mathcal{N}(L(E_{ij}),I_k)}\right|dx\\
&= \frac{1}{n^2}\sum_{i,j} d_{TV}(\mathcal{D}_1,\mathcal{N}(L(E_{ij}),I_k))\\
&= \frac{1}{n^2}\sum_{i,j} d_{TV}(\mathcal{N}(0,I_k),\mathcal{N}(L(E_{ij}),I_k))\\
&\le \frac{1}{n^2}\sum_{i,j} C'\alpha\eta_p \| L_{*, ij} \|_2 &\text{[from \pref{lem:gaussianTVdistance}]}\\
&= \frac{C'\alpha\eta_p}{n^2}\|L\|_{1,2}\\
&\le \frac{C'\alpha\eta_p}{n^2}\cdot n \|L\|_F = C'\alpha\eta_p \cdot \frac{\sqrt{k}}{n} &\text{[by Cauchy-Schwarz]} 
\end{align*}
\end{proof}

\subsection{Missing Proofs from \pref{sec:qsmall_plarge}}\label{app:qsmall_plarge}
\begin{proof}[Proof of \pref{lem:qsmallplargesep}]
We claim that for a diagonal matrix $D$, $\arg\max_{\|x\|_q=1}\|Dx\|_p$ is achieved when $x$ is one of the $e_i$ standard basis vectors
$e_i$. To see this,
\begin{align*}
\|Dx\|_p^p &= \sum_{i=1}^n |d_{ii}x_i|^p = \sum_{i=1}^n |d_{ii}|^p (|x_i|^q)^{p/q}
\leq \sum_{i=1}^n |d_{ii}|^p |x_i|^q
\leq \max_i |d_{ii}|^p
\end{align*}
which is achieved by picking $x=e_{i^*}$ where choice of $i=i^*$ maximizes $d_{ii}$.

Thus, to analyze the $q\to p$ norm of $G_1$, it suffices to analyze $\max_{x\in \{e_i\}}\|G_1x\|_p$, which is the same
as $\|g\|_{\infty}$ where $g$ is a vector of i.i.d. Gaussians. We can extract from the proof of \pref{lem:1toinftyseparation}
that $\|g\|_{\infty}$ is upper bounded by $\beta\sqrt{\log n}$ with probability at least $1-\frac{1}{n^2}$.

On the other hand, if the perturbation is at index $(i,i)$ and we pick $\alpha=\kappa(\beta+1)$, then $\|G_2e_i\|_p$ is at least
$\kappa\beta\sqrt{\log n}$ with probability at least $1-\frac{1}{n^2}$ implying the desired separation.
\end{proof}

\section{General approximation factors $\alpha$}\label{sec:generalapproxappendix}
\subsection{Sketching Matrix Construction and Upper Bounds}
Let us first define our sketch and then analyze its performance. For the sketch $S$, we group the rows of $A$ into $\frac{n}{\alpha^2}$
groups of size $\alpha^2$. We label the groups by $B_1, \dots, B_{n/\alpha^2}$ and let $\sigma_{1i}, \dots, \sigma_{\alpha^2 i}$ be
$\pm 1$ i.i.d random variables with equal probability for block $B_i$. Notice then that the $i^{\text{th}}$ row of $SA$ given by
$(SA)_{i, *}$ is: $$(SA)_{i, *} \triangleq \sum_{j \in B_i} \sigma_{ji} A_{i, *}$$
To analyze the performance of this sketch, we will need a helpful inequality describing the behavior of a random signed sums of reals.
\\
\begin{theorem}\label{thm:Khintchines}\textbf{Khintchine's Inequality} \cite{haagerup1981best}

Let $\{x_i\}_{i = 1}^n \in \R$ be reals and let $\{\bf{s}_i\}_{i = 1}^n$ be i.i.d $\pm 1$ random variables with equal probability and
let $0 < t < \infty$, we then have: 
\[A_p \sqrt{\sum_{i = 1}^n x_i^2} \leq \E{\left\lvert \sum_{i = 1}^n \bf{s}_ix_i \right\rvert ^p}^{1/p} \leq B_p \sqrt{\sum_{i = 1}^n x_i^2}\]
For some constants $A_p, B_p$ that only depend on $p$. 
\end{theorem}

Also recall that by Jensen's inequality, we can relate two norms of a vector $x \in \R^n$. 
\\
\begin{remark}\label{rem:Jensens}
For two positive reals, $p \geq q > 1$ and for a vector $x \in \R^n$ we have that: $\inorm{x}_p \leq n^{\frac{1}{q} - \frac{1}{p}}\inorm{x}_q$
\end{remark}

We then have the following theorems describing the sketching complexity of the sketch $S$ for $1 \leq p \leq 2$ and for $p > 2$. 
\\
\begin{theorem}\label{thm:between1and2}
For any $1 \leq p \leq 2$ and for the maximizer $x \in \R^n$ of $\inorm{A}_{q \to p}$ the sketch $S$ defined earlier where each block $B_i$ has size $B$ has the property that 
$$\Theta(1)\frac{1}{B^{1 - \frac{1}{p}}}\inorm{SAx}_p \leq \inorm{Ax}_p \leq \Theta(1)B^{\frac{1}{p} - \frac{1}{2}}\inorm{SAx}_p$$
with probability at least $\frac{99}{100}$
\end{theorem}

\begin{proof}
Let us first show the first inequality in the theorem statement.
\begin{align*}
\intertext{For some coordinate $1 \leq i \leq \frac{n}{B}$:}
\abs{(SAx)_i}^p &= \abs{\sum_{j \in B_i} \sigma_j(Ax)_j}^p \leq \brackets{\sum_{j \in B_i} |(Ax)_j|}^p \\ 
\intertext{By \pref{rem:Jensens} relating $\inorm{\cdot}_1$ and $\inorm{\cdot}_p$}
&\leq B^{p - 1} \sum_{j \in B_i} \abs{(Ax)_j}^p \\
\therefore \inorm{(SAx)_i}_p &= \brackets{\sum_{i = 1}^{n/B} \abs{(SAx)_i}^p }^{1/p} \leq B^{1 - \frac{1}{p}}\inorm{Ax}_p 
\end{align*}
Notice that the first inequality holds irrespective of the vector $x$, it holds for all vectors. Now let us show the second inequality of the theorem statement. 
\begin{align*}
\intertext{For some coordinate $1 \leq i \leq \frac{n}{B}$:}
\brackets{\sum_{j \in B_i} (Ax)_j^p}^{1/p} &\leq B^{\frac{1}{p} - \frac{1}{2}} \brackets{\sum_{j \in B_i} (Ax)_j^2}^{1/2}
&&\text{[By \pref{rem:Jensens}]} \quad [1]\\
&\leq \Theta(1)B^{\frac{1}{p} - \frac{1}{2}} \E{\abs{\sum_{j \in B_i}\sigma_j(Ax)_j}^p}^{1/p} &&\text{[By Khintchine's Ineq.]} \quad [2]\\ 
\therefore \sum_{i = 1}^{n/B}\sum_{j \in B_i} (Ax)_j^p &= \inorm{Ax}_p^p \leq \Theta(1)B^{p\brackets{\frac{1}{p} - \frac{1}{2}}}\E{\inorm{SAx}_p^p}
\end{align*}
Notice that the second inequality of the theorem statement follows by Markov's inequality. 

Notice that the success probability of line $[2]$ is constant for each block. To get constant success probability over the entire set of blocks, we construct $O(\log(n))$ i.i.d copies of each block $B_i$ given by $\{B_i^j\}_{i = 1}^{O(\log(n))}$. We then pick $j$ such that it is the index realizing the quantity $\mathsf{median}_{j \in [O(\log(n))]} \|(S_jAx)_i\|_p$ where $S_j$ corresponds the sketch with the $j^{\text{th}}$ copy of the blocks. Then, by standard concentration bounds, we can get $1 - \frac{1}{n/B}$ success probability for each set of blocks $B_i$ and then union bound over the $\frac{n}{B}$ blocks giving us constant success probability. 
\end{proof}
\hfill
\begin{theorem}\label{thm:above2}
For any $p > 2$ and for the maximizer $x \in \R^n$ of $\inorm{A}_{q \to p}$ the sketch $S$ defined earlier where each block $B_i$ has size $B$ has the property that 
$$\Theta(1)\frac{1}{B^{1 - \frac{1}{p}}}\inorm{SAx}_p \leq \inorm{Ax}_p \leq \Theta(1)\inorm{SAx}_p$$
\end{theorem}
The proof for \pref{thm:above2} is the same as that for \pref{thm:between1and2} except that there is no dilation while upper bounding the $\inorm{Ax}_p$ with the 2-norm in line $[1]$ of the proof. 

Notice that the above theorems imply that the sketch $S$ is a $\sqrt{B}$-approximation when $0 \leq p \leq 2$ and a $B^{1 - \frac{1}{p}}$-approximation when $p > 2$ because it states that the sketch is stretching $\inorm{Ax}_p^p$ by at most some factor and dilating it by at most some factor and hence the approximation ratio is simply the product of these factors.

\end{document}